\documentclass[12pt]{elsarticle}

\usepackage{geometry}
\usepackage{amsfonts,amssymb,amsmath,amsthm}
\usepackage[small]{complexity}
\usepackage{microtype}
\usepackage{enumitem}
\usepackage{hyperref} 
\usepackage{bbm}
\usepackage{url}

\usepackage{tikz}
\usetikzlibrary{arrows,shapes}

\setcitestyle{numbers,sort&compress}

\newtheorem{theorem}{Theorem}
\newtheorem{lemma}[theorem]{Lemma}
\newtheorem{proposition}[theorem]{Proposition}
\newtheorem{corollary}[theorem]{Corollary}

\theoremstyle{definition}
\newtheorem*{definition}{Definition}

\newtheorem*{claim}{Claim}
\newtheorem*{myquestion}{Direct sum question}

\theoremstyle{remark}
\newtheorem*{remark}{Remark}

\newtheorem{fact}{Fact}

\newcommand{\N}{\mathbb{N}} 
\newcommand{\F}{\mathbb{F}} 

\newcommand{\Ci}{\ensuremath{\mathcal{C}}} 
\newcommand{\Oh}{\ensuremath{\tilde{O}}}

\renewcommand{\A}{{\boldsymbol A}}
\newcommand{\B}{{\boldsymbol B}}
\renewcommand{\R}{{\boldsymbol R}}

\newcommand{\OR} {\texorpdfstring{\ensuremath{\lor}}{Or}}
\newcommand{\SUM}{\texorpdfstring{\ensuremath{+}}{Sum}}
\newcommand{\XOR}{\texorpdfstring{\ensuremath{\oplus}}{Xor}}

\newclass{\X}{X}
\newclass{\Y}{Y}

\newcommand{\Rewrite}{{\sffamily\upshape Rewrite}}
\newclass{\CNFSAT}{CNF\text{-}SAT}
\newclass{\UniqueSAT}{Unique\text{-}SAT}

\newcommand{\NumSAT}{{\raisebox{.05em}{\upshape\#}\CNFSAT}}
\newcommand{\ParitySAT}{{\raisebox{.057em}{\XOR}\CNFSAT}}

\newcommand{\xrewrite}{\texorpdfstring{{\upshape$(\OR,\XOR)$-\Rewrite}}{OR/XOR-Rewrite}}
\newcommand{\srewrite}{\texorpdfstring{{\upshape$(\OR,\SUM)$-\Rewrite}}{OR/SUM-Rewrite}}

\newcommand{\gap}{\textit{Gap}}

\newcommand{\card}[1]{\left\lvert {#1} \right\rvert}

\newcommand{\Prob}[1]{\Pr\left[ {#1} \right]}

\DeclareMathOperator{\GF}{GF}
\DeclareMathOperator{\rank}{rk}
\DeclareMathOperator{\supp}{supp}
\DeclareMathOperator{\cosupp}{co-supp}
\renewcommand{\poly}{\operatorname{poly}}

\hypersetup{
    colorlinks=true,
    pdfauthor={Magnus Find, Mika G\"o\"os, Matti J\"arvisalo, Petteri Kaski, Mikko Koivisto, Janne H. Korhonen},
    pdftitle={Separating OR, SUM, and XOR Circuits},
}

\makeatletter
\def\ps@pprintTitle{%
  \let\@oddhead\@empty
  \let\@evenhead\@empty
  \def\@oddfoot{\reset@font\hfil\thepage\hfil}
  \let\@evenfoot\@oddfoot
}
\makeatother

\begin{document}

\begin{frontmatter}

\title{{\bf Separating OR, SUM, and XOR Circuits}\tnoteref{conf}}
\tnotetext[conf]{This work is an extended version of two preliminary conference abstracts \cite{find13separating,jarvisalo12finding}.}

\author[usd]{Magnus Find}          
\author[uoft,hy]{Mika G\"o\"os}    
\author[hy]{Matti J\"arvisalo}     
\author[aalto]{Petteri Kaski}      
\author[hy]{\\Mikko Koivisto}        
\author[hy]{Janne H. Korhonen}     

\address[usd]{Department of Mathematics and Computer Science, University of Southern Denmark, Denmark}
\address[uoft]{Department of Computer Science, University of Toronto, Canada}
\address[hy]{HIIT \& Department of Computer Science, University of Helsinki, Finland}
\address[aalto]{HIIT \& Department of Information and Computer Science, Aalto University,  Finland}

\begin{abstract}
Given a boolean $n$ by $n$ matrix $A$ we consider arithmetic circuits for computing the transformation $x\mapsto Ax$ over different semirings. Namely, we study three circuit models: monotone OR-circuits, monotone SUM-circuits (addition of non-negative integers), and non-monotone XOR-circuits (addition modulo 2). Our focus is on \emph{separating} these models in terms of their circuit complexities. We give three results towards this goal:
\begin{enumerate}[label=(\arabic*),noitemsep]
\item We prove a direct sum type theorem on the monotone complexity of tensor product matrices. As a corollary, we obtain matrices that admit OR-circuits of size $O(n)$, but require SUM-circuits of size $\Omega(n^{3/2}/\log^2n)$.
\item We construct so-called \emph{$k$-uniform} matrices that admit XOR-circuits of size $O(n)$, but require OR-circuits of size $\Omega(n^2/\log^2n)$.
\item We consider the task of \emph{rewriting} a given OR-circuit as a XOR-circuit and prove that any subquadratic-time algorithm for this task violates the strong exponential time hypothesis.
\end{enumerate}
\end{abstract}


\end{frontmatter}

\newpage
\section{Introduction}

A basic question in arithmetic complexity is to determine the minimum size of an arithmetic circuit that evaluates a linear map $x\mapsto Ax$. In this work we approach this question from the perspective of relative complexity by varying the circuit model while keeping the matrix $A$ fixed, with the goal of separating different circuit models. That is, our goal is to show the existence of $A$ that admit small circuits in one model but have only large circuits in a different model. 

We will focus on boolean arithmetic and the following three circuit models. 
Our circuits consist of either
\begin{enumerate}[noitemsep]
\item only \OR-gates (i.e., boolean sums; rectifier circuits),
\item only \SUM-gates (i.e., integer addition; cancellation-free circuits), or
\item only \XOR-gates (i.e., integer addition mod 2).
\end{enumerate}
These three types of circuits have been studied extensively in their own right (see Section~\ref{sec:prior-work}), but fairly little is known about their relative powers.

Each model admits a natural description both from an algebraic and a combinatorial perspective. 

\paragraph{Algebraic perspective}
In the three models under consideration, each circuit with inputs $x_1,\ldots,x_n$ and outputs $y_1,\ldots,y_m$ computes a vector of \emph{linear forms}
\[
y_i=\sum_{j=1}^n a_{ij}x_j,\qquad i=1,\ldots,m.
\]
That is, $y=Ax$, where $A=(a_{ij})$ is an $m$ by $n$ boolean matrix with $a_{ij}\in\{0,1\}$ and the arithmetic is either 
\begin{enumerate}[noitemsep]
\item
in the boolean semiring $(\{0,1\},\lor,\land)$, 
\item
in the semiring of non-negative integers $(\N,+,\cdot)$, or 
\item
in $\GF(2)$.
\end{enumerate}
As an example, Fig.~\ref{fig:gap-example} displays two circuits for computing $y=Ax$ for the same $A$ using two different operators; the circuit on the right requires one more gate.
\begin{figure}[!t]
\begin{center}
\includegraphics[width=6.6cm]{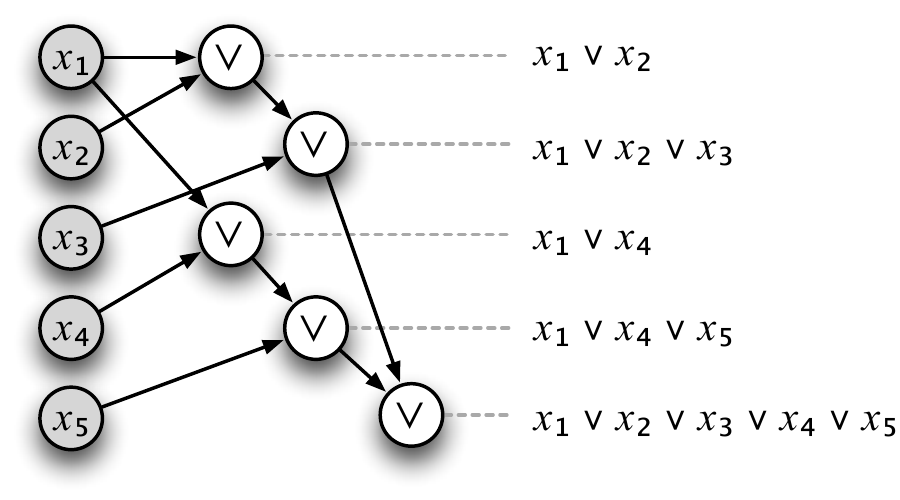}
\includegraphics[width=6.6cm]{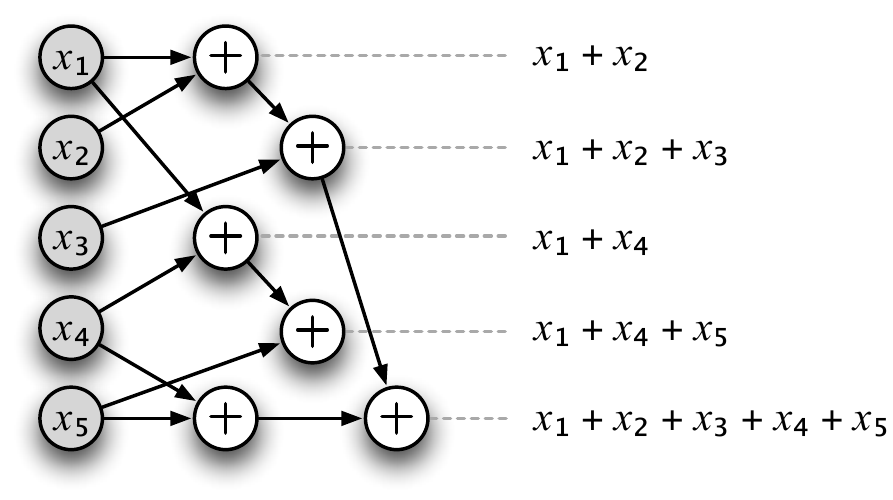}
\end{center}
\caption{An \OR-circuit (left) and a \SUM-circuit (right).}
\label{fig:gap-example}
\end{figure}

\paragraph{Combinatorial perspective}
A circuit computing $y=Ax$ for a boolean matrix $A$ can also be viewed combinatorially: every gate $g$ is associated with a subset of the formal variables $\{x_1,\ldots,x_n\}$; this set is called the \emph{support} of $g$ and it is denoted $\supp(g)$. The input gates correspond to the singletons $\{x_j\}$, $j=1,\ldots,n$, and every non-input gate computes either
\begin{enumerate}[noitemsep]
\item the set union (\OR),
\item the disjoint set union (\SUM), or
\item the symmetric difference (\XOR) of its children.
\end{enumerate}
This way an output gate $y_i$ will have $\supp(y_i)= \{x_j : a_{ij} = 1\}$.

Note the special structure of a \SUM-circuit: there is at most one directed path from any input $x_j$ to any output $y_i$. In fact, from this perspective, every \SUM-circuit for $A$ is easy to interpret both as an \OR-circuit for $A$, and as a \XOR-circuit for $A$ (equivalently, there are onto homomorphisms from $(\N,+,\cdot)$ to $(\{0,1\},\lor,\land)$ and $\GF(2)$). In this sense, both \OR- and \XOR-circuits are at least as efficient as \SUM-circuits.

\paragraph{Relative complexity}
More generally we fix a boolean matrix $A$ and ask how the circuit complexity of computing $y=Ax$ depends on the underlying arithmetic. 

To make this quantitative, denote by $C_\OR(A)$, $C_\SUM(A)$, and $C_\XOR(A)$ the minimum number of wires in an unbounded fan-in circuit for computing $y=Ax$ in the respective models. For simplicity, we restrict our attention to the case of square matrices so that $m=n$.

For $\X,\Y\in\{\OR,\SUM,\XOR\}$, we are interested in the complexity ratios
\[
\gap_{\X/\Y}(n)\ := \max_{A\in\{0,1\}^{n\times n}} C_\X(A)/C_\Y(A).
\]
For example, we have that $\gap_{\OR/\SUM}(n)=\gap_{\XOR/\SUM}(n)= 1$ and that $\gap_{\SUM/\XOR}(n)\geq \gap_{\OR/\XOR}(n)$ for all $n$, by the above fact that each \SUM-circuit can be interpreted as an \OR-circuit and as a \XOR-circuit.

We review the motivation for studying separation bounds in Section~\ref{sec:prior-work}. Next, we state our results, which are summarised in Figure~\ref{fig:bounds}.

\begin{figure}[t]
\begin{center}
\begin{tikzpicture}[auto,scale=1.2,thick]
\tikzstyle{n} = [fill=white,font=\large,draw=black!50,thick,
rectangle, rounded corners,
minimum height=1.9em,
minimum width=1.9em,
inner sep=0];

\small
\node[n] (sum) at (90:2) {{$+$}};
\node[n] (xor) at (-20:2) {{$\oplus$}};
\node[n] (or) at (200:2) {{$\lor$}};
\draw[-latex,bend left=10] (or) edge node [left] {
$\begin{array}{c}
\Omega(n^{1/2}/\log^2n), \\
{}[n^{1-o(1)}]
\end{array}$} (sum);
\draw[-latex,bend left=10] (sum) edge node [right] {1} (or);
\draw[-latex,bend right=10] (sum) edge node [left] {1} (xor);
\draw[-latex,bend left=10] (or) edge node [above] {$[n^{1-o(1)}]$} (xor);
\draw[-latex,bend left=10] (xor) edge node [below] {$\Omega(n/\log^2n)$} (or);
\draw[-latex,bend right=10] (xor) edge node [right] {$\ \Omega(n/\log^2n)$} (sum);

\end{tikzpicture}
\end{center}
\vspace*{-4mm}
\caption{Separation bounds. An arrow from $\Y$ to $\X$ is labelled with $\gap_{\X/\Y}(n)$; bounds for $(\X,\Y)$-\Rewrite{} are given inside square brackets.}
\label{fig:bounds}
\end{figure}
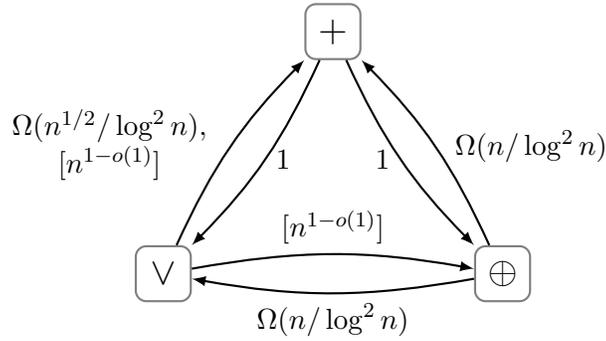

\subsection{Our results} \label{sec:results}

We begin by studying the monotone complexity of tensor product matrices of the form
\[
A=B_1\otimes B_2,
\]
where $\otimes$ denotes the usual Kronecker product of matrices. In Section~\ref{sec:tensor}, we prove a direct sum type theorem on their monotone complexity. As a corollary, we obtain matrices that are easy for \OR-circuits, $C_\OR(A)=O(n)$, but hard for \SUM-circuits, $C_\SUM(A)=\Omega(n^{3/2}/\log^2n)$. This implies our first separation:
\begin{theorem} \label{thm:tensor}
$\gap_{\SUM/\OR}(n) = \Omega(n^{1/2}/\log^2n)$.
\end{theorem}
We are not aware of any prior lower bound techniques that work against \SUM-circuits, but not against \OR-circuits. Hence, as far as we know, Theorem~\ref{thm:tensor} is a first step in this direction.

Next, we separate \OR- and \SUM-circuits from \XOR-circuits by considering matrices that look \emph{locally random} in the following sense:
\begin{definition}[$k$-uniformity] \label{def:uniformity}
A random matrix $\A$ is called \emph{$k$-uniform} if the entries in every $k\times k$ submatrix have a marginal distribution that is uniform on~$\{0,1\}^{k\times k}$.
\end{definition}
Equivalently, a matrix is $k$-uniform if each of its entries is $0$ or $1$ with equal probability and the entries in every $k\times k$ submatrix are mutually independent.

In Section~\ref{sec:or-xor-sep} we construct $n^{\Omega(1)}$-uniform matrices that are easy for $\XOR$-circuits:
\begin{theorem} \label{thm:k-uniform}
There are $n^{\Omega(1)}$-uniform matrices $\A$ having $C_\XOR(\A) = O(n)$.
\end{theorem}

These $k$-uniform matrices turn out to be difficult to compute using monotone circuits. Indeed, as a corollary, we will obtain our second separation:
\begin{corollary} \label{cor:or-xor-sep}
$\gap_{\OR/\XOR}(n),\gap_{\SUM/\XOR}(n) = \Omega(n/\log^2n)$.
\end{corollary}
Separations between \OR- and \XOR-circuits have also been considered by Sergeev et al.~\cite{gashkov11complexity,grinchuk11thin} who proved the slightly weaker bound $\gap_{\OR/\XOR}(n) = \Omega(n/(\log^6n\log\log n))$. Furthermore, Jukna~\cite{jukna13comment} has informed us that the bound in Corollary~\ref{cor:or-xor-sep} can actually be proved more directly using existing methods~\cite{jukna06disproving,pudlak04pseudorandom}. Nevertheless, we hope our alternative approach via $k$-uniform matrices might be of independent interest---for example, in closing the gap between the current lower bound $\gap_{\OR/\XOR}(n) = \Omega(n/\log^2 n)$ and the best known upper bound $\gap_{\OR/\XOR}(n) = O(n/\log n)$; see Section~\ref{sec:prior-work}.

As is true in the case of $\gap_{\OR/\XOR}$ we conjecture more generally that all the non-trivial complexity gaps between the three models are of order $n^{1-o(1)}$. While we are unable to enlarge the gap in Theorem \ref{thm:tensor}, or prove any super-constant lower bounds on $\rho_{\XOR/\OR}$, our final result provides some evidence towards these conjectures.

In Section~\ref{sec:rewriting}, we show that if certain \OR-circuits that are derived from CNF formulas could be efficiently \emph{rewritten} as equivalent \SUM- or \XOR-circuits, this would imply unexpected consequences for exponential-time algorithms. More precisely, we study the following problem.
\begin{description}
\item[The {\normalfont$(\X,\Y)$-\Rewrite} problem:] On input an \X-circuit $\Ci$, output a \Y-circuit that computes the same matrix as $\Ci$. 
\end{description}
Both \srewrite{} and \xrewrite{} admit simple algorithms that output a circuit of size $O(|\Ci|^2)$ in time $O(|\Ci|^2)$. However, we show that any significant improvement on these algorithms would give a non-trivial $2^{(1 - \epsilon)n}\poly(n,m)$ time algorithm for deciding whether an $n$-variable $m$-clause CNF formula is satisfiable---this violates \emph{the strong exponential time hypothesis}~\cite{Impagliazzo:ksat}:

\begin{theorem} \label{thm:rewrite}
Neither \srewrite{} nor \xrewrite{} can be solved in time $O(\card{\Ci}^{2-\epsilon})$ for any constant $\epsilon > 0$, unless the strong exponential time hypothesis fails.
\end{theorem}

Theorem \ref{thm:rewrite} provides evidence, e.g., for the conjecture $\rho_{\XOR/\OR}=n^{1-o(1)}$ in the following sense. If there is a family of matrices $A$ witnessing $C_\XOR(A)/C_\OR(A)=n^{1-o(1)}$, then clearly no $O(|\Ci|^{2-\epsilon})$-time algorithm exists for \xrewrite{}: if we are given a minimum-size \OR-circuit for $A$ as input, there is no time to write down a legal output.

Our proof of Theorem \ref{thm:rewrite} shows, in particular, that an $O(\card{\Ci}^{2-\epsilon})$-time algorithm for \srewrite{} would give an improved algorithm for counting the number of satisfying assignments to a given CNF formula (\NumSAT). Similarly, an $O(\card{\Ci}^{2-\epsilon})$-time algorithm for \xrewrite{} would give an improved algorithm for deciding whether the number of satisfying assignments is odd (\ParitySAT).

\subsection{Notation}\label{sec:notation}

A {\em circuit} $\Ci$ is a directed acyclic graph where the vertices of in-degree (or {\em fan-in}) zero are called {\em input gates} and all other vertices are called {\em arithmetic gates}. One or more arithmetic gates are designated as {\em output gates}. The size $|\Ci|$ of the circuit is the number of edges (or {\em wires}) in the circuit.

We abbreviate $[n]:=\{1,\ldots,n\}$; all our logarithms are to base $2$ by default; and we write random variables in boldface.

\section{Related work} \label{sec:prior-work}

\paragraph{Upper bounds}
The trivial depth-$1$ circuit for a boolean matrix $A$ uses $\card{A}$ wires, where we denote by $\card{A}$ the \emph{weight} of $A$, i.e., the number of 1-entries in $A$. Even though $\card{A}$ might be of order $\Theta(n^2)$, Lupanov (as presented by Jukna~\cite[Lemma 1.2]{jukna12boolean}) constructs depth-2 circuits (applicable in all the three models) of size $O(n^2/\log n)$ for any $A$. This implies the universal upper bound
\[
\gap_{\X/\Y}(n)=O(n/\log n). \tag{Lupanov}
\]

\paragraph{Lower bounds}
Standard counting arguments~\cite[\defaultS1.4]{jukna12boolean} show that most $n \times n$ matrices have wire complexity $\Omega(n^2/\log n)$ in each of the three models. Combining this with Lupanov's upper bound we conclude that a random matrix does little to separate our models:
\begin{fact}\label{fact:random-a}
For a uniformly random $\A$, the ratio $C_\X(\A)/C_\Y(\A)$ is a constant w.h.p.
\end{fact}
Unsurprisingly, it can also be shown that finding a minimum-size circuit for a given matrix is NP-hard in all the models. For \OR- and \SUM-circuits this follows from the NP-completeness of the {\sf Ensemble Computation} problem as defined by Garey and Johnson~\cite[PO9]{garey-johnson}. For \XOR-circuits this was proved by Boyar et al.~\cite{boyar13logic}.

\paragraph{\OR-circuits}
The study of \OR-circuits (sometimes called \emph{rectifier circuits}) has been centered around finding \emph{explicit} matrices that are hard for \OR-circuits. Here, dense \emph{rectangle-free} matrices and their generalisations, \emph{$(s,t)$-free} matrices, are a major source of lower bounds.
\begin{definition}
A matrix $A$ is called \emph{$(s,t)$-free} if it does not contain an $(s+1)\times(t+1)$ all-1 submatrix. Moreover, $A$ is simply called \emph{$k$-free} if it is $(k,k)$-free.
\end{definition}
Nechiporuk~\cite{nechiporuk71boolean} and independently Lamagna and Savage~\cite{lamagna74computational} constructed the first examples of dense $1$-free matrices $A$ achieving $C_\OR(A)=\Omega(n^{3/2})$. Subsequently, Mehlhorn~\cite{mehlhorn79some} and Pippenger~\cite{pippenger80another} established the following theorem that gives a general template for this type of lower bound; we use it extensively later.
\begin{theorem}[Mehlhorn--Pippenger] \label{thm:mehlhorn-pippenger}
If $A$ is $(s,t)$-free, then $C_\OR(A)\geq |A|/(st)$.
\end{theorem}
Currently, the best lower bound for an explicit $A$ is obtained by applying Theorem \ref{thm:mehlhorn-pippenger} to a matrix construction of Koll\'{a}r et al.~\cite{kollar96norm}; the lower bound is $C_\OR(A)\geq n^{2-o(1)}$ (see also Gashkov and Sergeev~\cite[\defaultS3.2]{gashkov11complexity}).

\paragraph{\XOR-circuits}
It is a long-standing open problem to exhibit explicit matrices requiring super-linear size \XOR-circuits. No such lower bounds are known even for log-depth circuits, and the only successes are in the case of bounded depth \cite{alon90linear,gal12tight}, \cite[\defaultS13.5]{jukna12boolean}. This, together with Fact~\ref{fact:random-a}, makes it particularly difficult to prove lower bounds on $\gap_{\XOR/\OR}$.



\paragraph{\SUM-circuits}
Additive circuits have been studied extensively in the context of the \emph{addition chain} problem (see Knuth~\cite[\defaultS4.6.3]{knuth98art} for a survey) and its generalisations~\cite{pippenger80evaluation}.

In cryptography, as observed by Boyar et al.~\cite{boyar13logic}, many heuristics that have been proposed for finding small \XOR-circuits produce, in fact, \SUM-circuits that do not exploit the cancellation of variables that is available in $\GF(2)$. Thus, the measure $\gap_{\SUM/\XOR}$ gives a lower bound on the approximation ratio achieved by any such minimisation heuristic.

\paragraph{Algebraic complexity}
A particular motivation for studying the separation between \OR- and \SUM-circuits is to understand the complexity of zeta transforms on partial orders~\cite{Bjorklund:zeta}. Indeed, the characteristic matrix of every partial order $\leq$ has an \OR-circuit proportional to the number of covering pairs in $\leq$, but the existence of small \SUM-circuits (and hence fast zeta transforms) is not currently understood satisfactorily.

\paragraph{Strong exponential time hypothesis}
Theorem \ref{thm:rewrite} is similar to other recent lower bound results for polynomial-time solvable problems based on the strong exponential time hypothesis \cite{DBLP:conf/soda/PatrascuW10}. See also \cite{cygan2012}.

\section{\SUM/\OR-Separation} \label{sec:tensor}

In this section we give a direct sum type theorem for the monotone complexity of tensor product matrices. Using this, we obtain a separation of the form
\begin{equation} \label{eq:tensor-sep}
\begin{aligned}
C_\OR(B\otimes A)  &= O(N), \\
C_\SUM(B\otimes A) &= \Omega(N^{3/2}/\log^2 N),
\end{aligned}
\end{equation}
where $\otimes$ denotes the usual Kronecker product of matrices and $N=n^2$ denotes the number of input and output variables. This will prove Theorem~\ref{thm:tensor}.

\subsection{Tensor products}
As a first example, let $A$ be a fixed boolean $n\times n$ matrix and consider the matrix product
\begin{equation} \label{eq:mult}
X\mapsto AX\,,
\end{equation}
where we think of $X$ as a matrix of $N = n\times n$ input variables. If we arrange these variables into a column vector $x$ by stacking the columns of $X$ on top of one another, then (\ref{eq:mult}) becomes
\begin{equation} \label{eq:tensormult}
x \mapsto (I\otimes A)x,
\end{equation}
where $I$ is the $n\times n$ identity matrix. That is, $I\otimes A$ is the block matrix having $n$ copies of $A$ on the diagonal.

The transformation (\ref{eq:tensormult}) famously admits non-trivial \XOR-circuits due to the fact that fast matrix multiplication algorithms can be expressed as small bilinear circuits over $\GF(2)$. However, it is easy to see that in the case of our monotone models, no non-trivial speed-up is possible: any \OR-circuit for (\ref{eq:tensormult}) must compute $A$ independently $n$ times:
\begin{equation} \label{eq:mult-lb}
C_\OR(I\otimes A) = n\cdot C_\OR(A).
\end{equation}
This follows from the observation that two subcircuits corresponding to two different columns of $X$ cannot share gates due to monotonicity.

\paragraph{Our approach}
We will generalise the above setting slightly and use tensor products of the form $B\otimes A$ to separate \OR- and \SUM-circuits. Analogously to (\ref{eq:mult}), one can check that the matrix $B\otimes A$ corresponds to computing the mapping
\begin{equation}\label{eq:axb}
X\mapsto AXB^\top.
\end{equation}
We aim to show that for suitable choices of $A$ and $B$ computing $B\otimes A$ is easy for \OR-circuits but hard for \SUM-circuits.  We will choose $A$ to have large complexity (e.g., choose $A$ at random), and think of $B$ as dictating how many independent copies of $A$ a circuit must compute.

More precisely, define $\rank_\OR(B)$ and $\rank_\SUM(B)$ as the minimum $r$ such that $B$ can be written as $B=PQ^\top$ over the boolean semiring or over the semiring of non-negative integers, respectively, where $P$ and $Q$ are $n\times r$ matrices. Equivalently, $\rank_\OR(B)$ (resp., $\rank_\SUM(B)$) is the minimum number of rectangles (resp., non-overlapping rectangles) that are required to cover all $1$-entries of $B$.

These cover numbers appear often in the study of communication complexity~\cite{kushilevitz97communication}. In this context, the matrix $B=\bar{I}$---the boolean complement of the identity $I$---is the usual example demonstrating a large gap between the two concepts~\cite[Example 2.5]{kushilevitz97communication}:
\begin{align*}
\rank_\OR(\bar{I}) &= \Theta(\log n),\\
\rank_\SUM(\bar{I}) &= n.
\end{align*}
We will use this gap to show that, up to polylogarithmic factors,
\begin{align*}
C_\OR(\bar{I}\otimes A)\enspace &\approx\enspace \rank_\OR(\bar{I})\cdot n^2,\\
C_\SUM(\bar{I}\otimes A)\enspace &\approx\enspace \rank_\SUM(\bar{I})\cdot n^2.
\end{align*}
In terms of the number of input variables $N=n^2$, we will obtain~(\ref{eq:tensor-sep}).
\subsection{Upper bound for \OR-circuits}

Suppose $B=PQ^\top$ where $P$ and $Q$ are $n\times \rank_\OR(B)$ matrices. We can compute (\ref{eq:axb}) as
\[
(A(XQ))P^\top,
\]
which requires 3 matrix multiplications, each involving $\rank_\OR(B)$ as one of the dimensions (the other dimensions being at most $n$).

If these 3 multiplications are naively implemented with an \OR-circuit of depth 3, each layer will contain at most $\rank_\OR(B)n^2$ wires so that $C_\OR(B\otimes A) \leq 3\rank_\OR(B)n^2$. However, one can still use Lupanov's techniques to save an additional logarithmic factor: if $\rank_\OR(B)=O(\log n)$, Corollary~1.35 in Jukna \cite{jukna12boolean} can be applied to show that each of the three multiplications above can be computed using $O(n^2)$ wires. Thus, for $B=\bar{I}$ we get
\begin{lemma} \label{lemma:upper-bound}
$C_\OR(\bar{I}\otimes A) = O(n^2)$ for all $A$. \qed
\end{lemma}

\subsection{Lower bound for \SUM-circuits}

Intuitively, since low-rank decompositions are not available for $\bar{I}$ in the semiring of non-negative integers, a \SUM-circuit for $\bar{I}\otimes A$ should be forced to compute $\rank_\SUM(\bar{I})=n$ independent copies of $A$. More generally, we ask
\begin{myquestion}
Do we have $C_\SUM(B\otimes A) \geq \rank_\SUM(B) \cdot C_\SUM(A)$ for all $A$, $B$?
\end{myquestion}

Alas, we can answer this affirmatively only in some special cases. For example, the trivial case $B=I$ was discussed above (\ref{eq:mult-lb}), and it is not hard to generalise the argument to show that the lower bound holds in case $B$ admits a fooling set of size $\rank_\SUM(B)$. (When $B$ is viewed as an incidence matrix of a bipartite graph, a \emph{fooling set} is a matching no two of whose edges induce a 4-cycle. See~\cite[\defaultS1.3]{kushilevitz97communication}.) However, since this will not be the case when $B=\bar{I}$, we will settle for the following version, which suffices for the separation result.
\begin{theorem} \label{thm:st-free-lb}
For all $(s,t)$-free $A$,
\begin{equation}
C_\SUM(B\otimes A) \geq \rank_\SUM(B) \cdot\frac{|A|}{st}.
\label{eq:lb}
\end{equation}
\end{theorem}
Note that if we set $B=I$ in Theorem \ref{thm:st-free-lb} we recover essentially Theorem~\ref{thm:mehlhorn-pippenger}.

For the purposes of the proof we switch to the combinatorial perspective: For $A$ and $B$ we introduce two sets of $n$ formal variables $X_A$ and $X_B$. Moreover, we let $A_1,\ldots,A_n\subseteq X_A$ and $B_1,\ldots,B_n\subseteq X_B$ denote the associated outputs. That is, each output $A_i$ is defined by one row of $A$, and each output $B_j$ is defined by one row of $B$. With this terminology, the input variables for $B\otimes A$ are the pairs in $X_A\times X_B$; we think of $X_A$ as indexing the rows and $X_B$ as indexing columns of the variable matrix $X_A\times X_B$. Finally, $B\otimes A$ corresponds to computing the $n^2$ outputs
\[
A_i\times B_j,\quad\text{for } i,j\in [n].
\]

In the following proof we use the $(s,t)$-freeness of $A$ to ``zoom in'' on that layer of the circuit which reveals the large wire complexity (similarly to Mehlhorn~\cite{mehlhorn79some}).  We advise the reader to first consider the case $s=t=1$, as this already contains the main idea of the proof. 
\begin{proof}[Theorem~\ref{thm:st-free-lb}]
Let $\Ci$ be a \SUM-circuit computing $B\otimes A$. As a first step, we simplify $\Ci$ by allowing input gates to have larger-than-singleton supports. Namely, let $F$ consist of those gates of $\Ci$ whose supports are contained in a \emph{$t$-wide row cylinder} of the form $Y\times X_B$ where $Y\subseteq X_A$ and $|Y|\leq t$. We simply declare that all computations done by gates in $F$ come for free: we promote a gate in $F$ to an input gate and delete all its incoming wires. We continue to denote the modified circuit by $\Ci$---clearly, these modifications only decrease its wire complexity.

Call a wire that is connected to an input gate an \emph{input wire} and denote the set of input wires by $W$. The wire complexity lower bound (\ref{eq:lb}) will follow already from counting the number $|W|$ of input wires.

For $i\in[n]$ denote by $\Ci_i$ the subcircuit of $\Ci$ computing the $n$ outputs $A_i\times B_j$, $j\in[n]$, and denote by $W(i)$ the input wires of $\Ci_i$; we claim that
\begin{equation} \label{eq:claim}
|W(i)| \geq \rank_\SUM(B) \cdot \frac{|A_i|}{t}.
\end{equation}
Before we prove (\ref{eq:claim}), we note how it implies the theorem. Each input wire $w\in W$ is feeding into a non-input gate having their support not contained in a $t$-wide row cylinder. Due to $(s,t)$-freeness of $A$ this means that $w$ can appear only in at most $s$ different $\Ci_i$. Thus, the sum $\sum_i|W(i)|$ counts $w$ at most $s$ times and, more generally, we have
\[
 |W| = \biggl|\bigcup_{i=1}^n W(i)\biggr| \geq \sum_{i=1}^n \frac{|W(i)|}{s},
\]
which implies (\ref{eq:lb}) given (\ref{eq:claim}).

\vspace{1em}
\noindent
\emph{Proof of (\ref{eq:claim}).}
Fix $i\in[n]$. If $A_i$ is empty the claim is trivial. Otherwise fix a variable $x\in A_i$ and consider the structure of $\Ci_i$ when restricted to the variables $\{x\}\times X_B$. Since this set of variables can be naturally identified with $X_B$ by ignoring the first coordinate, we can view $\Ci_i$ as computing a copy of $B$ on the variables $\{x\}\times X_B$.

Indeed, we define the \emph{$x$-support} $\supp_x(w)$ of an input wire $w\in W(i)$ to be the set of $y\in X_B$ such that the variable $(x,y)$ is contained in the support of $w$. (The support of $w$ is simply the support of the adjacent input gate.) Moreover, we let
\[
 W_x(i) := \{ w\in W(i) : \supp_x(w)\neq \varnothing\}.
\]
Put otherwise, $W_x(i)$ consists of the input wires that are used by $\Ci_i$ in computing a copy of $B$ on the variables $\{x\}\times X_B$.  Associate to each $w\in W_x(i)$ a rectangle
\[
 R_x(w) := \cosupp_x(w)\times\supp_x(w),
\]
where $\cosupp_x(w)$ is the set of $j\in[n]$ such that $w$ appears in the subcircuit $\Ci_{ij}$ of $\Ci_i$ that computes the output $A_i\times B_j$. Now, the crucial observation is that the collection of rectangles $\{R_x(w) : w\in W_x(i)\}$ is a non-overlapping cover of $B$, because $\Ci_i$ computes a copy of $B$ by taking disjoint unions of the supports $\{\supp_x(w) : w\in W_x(i)\}$. Therefore, we must have that
\begin{equation} \label{eq:cc-bound}
 |W_x(i)| \geq \rank_\SUM(B).
\end{equation}

To finish the proof, we note that a single input wire $w\in W(i)$, being $t$-wide, can only be contained in the sets $W_x(i)$ for at most $t$ different $x\in A_i$. Thus, the sum $\sum_x |W_x(i)|$ counts $w$ at most $t$ times and, more generally, we have
\[
 |W(i)|=\biggl|\bigcup_{x\in A_i} W_x(i)\,\biggr|\geq \sum_{x\in A_i} \frac{|W_x(i)|}{t},
\]
which implies (\ref{eq:claim}) given (\ref{eq:cc-bound}).
\end{proof}
As will be shortly discussed in Section~\ref{sec:uniform-motivation}, a random matrix $\A\in\{0,1\}^{n\times n}$ is $O(\log n)$-free and has weight $|\A|=\Theta(n^2)$ w.h.p. Using these facts we obtain the following corollary, which, together with Lemma~\ref{lemma:upper-bound}, proves Theorem~\ref{thm:tensor}.
\begin{corollary}
A random $\A$ satisfies $C_\SUM(\bar{I}\otimes \A) = \Omega(n^3/\log^2 n)$ w.h.p. \qed
\end{corollary}

\section{\OR/\XOR-Separation} \label{sec:or-xor-sep}

In this section we use the probabilistic method to construct $k$-uniform matrices $\A$ that, for large enough $k$, will witness the following complexity gap with high probability:
\begin{align*}
C_\XOR(\A) &= O(n), \\
C_\OR(\A)  &= \Omega(n^2/\log^2 n).
\end{align*}
In what follows, all matrix arithmetic will be over $\F=\GF(2)$.

\subsection{Motivation for \texorpdfstring{$k$}{k}-uniform matrices} \label{sec:uniform-motivation}

Suppose first that $\A\in\F^{n\times n}$ is a random matrix where each entry is drawn uniformly and independently from $\F$. The probability that $\A$ fails to be $(k-1)$-free can be bounded from above by taking the union bound over all possible $k\times k$ submatrices:
\begin{equation} \label{eq:up}
\Prob{\,\A\text{ is not $(k-1)$-free}\,}\leq\binom{n}{k}^2 2^{-k^2}.
\end{equation}
It is easy to check (and well-known in the context of random graphs~\cite[\defaultS11]{bollobas01random}) that for $k \geq 2\log n$ this quantity tends to 0 as $n\to\infty$.

Our key observation here is that the estimate (\ref{eq:up}) only uses the property that the entries in each $k\times k$ submatrix of $\A$ are mutually independent. Indeed, the above analysis holds even when $\A$ is only $k$-uniform for $k\geq 2\log n$. Thus, we have the following lemma.
\begin{lemma}\label{lem:uniform-implies-free}
If $\A$ is $k$-uniform for $k\geq 2\log n$, then w.h.p.,
\[
C_\OR(\A)=\Omega(n^2/\log^2 n).
\]
\end{lemma}
\begin{proof}
Any $2$-uniform matrix $\A$ has pairwise independent entries so that $|\A|=\Theta(n^2)$ w.h.p.\ by Chebyshev's inequality. On the other hand, the above discussion implies that $\A$ is $2\log n$-free w.h.p. Thus, the claim follows from Theorem~\ref{thm:mehlhorn-pippenger}.

\end{proof}

Corollary~\ref{cor:or-xor-sep} is a consequence of Lemma~\ref{lem:uniform-implies-free} and Theorem~\ref{thm:k-uniform}. Thus, our remaining goal in this section is to prove Theorem \ref{thm:k-uniform}. 

\subsection{Proof of Theorem~\ref{thm:k-uniform}}

Let $m:=O(\sqrt{n})$. To construct a $k$-uniform matrix $\A$ we start with an $m\times n$ matrix $P$ that satisfies the following two properties:
\begin{enumerate}[label=(\arabic*),noitemsep]
\item $P$ has linear \XOR-complexity, $C_\XOR(P)=O(n)$.
\item Each set of $k=n^{\Omega(1)}$ columns of $P$ are linearly independent.
\end{enumerate}
Miltersen~\cite{miltersen98error} shows that such $P$ can be obtained as submatrices of certain generating matrices of linear codes, e.g., those of Spielman~\cite{spielman96linear}.
\begin{theorem}[{Miltersen~\cite[Theorem 1.4]{miltersen98error}}] \label{thm:miltersen}
Let $D\subseteq\F^n$. There are $O(\log|D|)\times n$ matrices $P$ with $C_\XOR(P)=O(n)$ such that the mapping $x\mapsto Px$ is injective on $D$.
\end{theorem}
Indeed, let $D\subseteq\F^n$ be the set of vectors of Hamming weight at most $k$. Note that if $P$ is injective on $D$, then it clearly has property (2). We also have that $|D|\leq (n+1)^k$ and so $\log|D|= O(k\log n)$. Thus, if we set $k:= \sqrt{n}/\log n$, we can apply Theorem~\ref{thm:miltersen} to obtain our desired $m\times n$ matrix $P$.

We can now define
\[
\A:= P^\top \R P,
\]
where $\R\in \F^{m\times m}$ is a matrix chosen uniformly at random; note that
$C_\XOR(\R)\leq |\R| \leq m^2=O(n)$. If we compute $\A$ in three stages in the obvious way, we obtain
\[
C_\XOR(\A) \leq C_\XOR(P^\top)+C_\XOR(\R) + C_\XOR(P) = O(n),
\]
where we used the fact that $C_\XOR(P^\top) = O(C_\XOR(P))$---roughly, this follows from simply reversing the direction of the wires in a \XOR-circuit computing $P$ (see Jukna~\cite[p.\ 46]{jukna12boolean}).

It remains to show that $\A$ is $k$-uniform. In fact, since our definition of $\A$ is a generalisation of how $k$-wise independent variables are typically constructed~\cite[\defaultS15.2]{alon00probabilistic}, the proof of the following lemma is somewhat routine.
\begin{lemma}
$\A$ is $k$-uniform.
\end{lemma}
\begin{proof}
We need to show that each submatrix $\A_{I\times J}$, where $I,J\subseteq[n]$ and $|I|=|J|=k$, is uniformly distributed in $\F^{k\times k}$. Write
\[
\A_{I\times J} = {P_I}^\top \R P_J,
\]
where $P_K$ is the submatrix of $P$ consisting of the columns with indices in $K\subseteq[n]$.

\begin{claim}
$\B := \R P_J$ is uniformly distributed in $\F^{m\times k}$.
\end{claim}
\noindent\emph{Proof of Claim.}
Let $\B_i = \R_i P_J$ denote the $i$-th row of $\B$. The rows $\B_i$, $i\in [m]$, are mutually independent variables, since the variables $\R_i$, $i\in [m]$, are. Therefore it suffices to show that $\B_i$ is uniformly distributed in $\F^{1\times k}$ for each $i\in[m]$. 

To this end, fix $i\in[m]$; we show that all the outcomes $\B_i = y$ where $y\in \F^{1\times k}$ are equally likely. For any $y\in \F^{1\times k}$ there is a vector $x\in \F^{1\times m}$ with $xP_J=y$ since $P_J$ has linearly independent columns. Hence $\R_iP_J = y$ iff $(\R_i-x)P_J = 0$. But $\R_i-x$ is distributed the same as $\R_i$ so that $\Prob{\R_iP_J = y} = \Prob{\R_iP_J = 0}$ is independent of the choice of $y$, as desired.\hfill{\large$\diamond$}

\medskip

Finally, the same analysis as above demonstrates that $\A_{I\times J} = {P_I}^\top \B$ is uniformly distributed in $\F^{k\times k}$ proving the lemma.
\end{proof}

\begin{remark}
Interestingly, Theorem~\ref{thm:mehlhorn-pippenger} is unable to prove a better lower bound than $C_\OR(A)=\Omega(n^2/\log^2 n)$ for any matrix $A$. Is it true that for every $n^{\Omega(1)}$-uniform $\A$, we have that $C_\OR(\A)=\Theta(n^2/\log n)$ w.h.p.? A positive answer would give the tight bound $\gap_{\OR/\XOR}(n) = \Theta(n/\log n)$.
\end{remark}

\section{Rewriting} \label{sec:rewriting}

In this section we study what would happen if \srewrite{} or \xrewrite{} could be solved in subquadratic time. Namely, we show that this eventuality would contradict the strong exponential time hypothesis. This will prove Theorem~\ref{thm:rewrite}. As discussed in Section~\ref{sec:results}, we interpret this as evidence for our conjectures $\rho_{\SUM/\OR}=n^{1-o(1)}$ and $\rho_{\XOR/\OR}=n^{1-o(1)}$.

\subsection{Preliminaries}
For purposes of computations, we tacitly assume that $|\Ci|\geq n$ for any $n$-input circuit $\Ci$ considered in this section. This is to make each $\Ci$ admit a binary representation of length $\Oh(|\Ci|)$ where the $\Oh$ notation hides factors polylogarithmic in $n$. For concreteness, $\Ci$ might be represented as two lists: (i) the list of gates in $\Ci$, with output gates indicated, and (ii) the list of wires in $\Ci$; both lists are given in topological order, with the input wires of each gate forming a consecutive sublist of the list of wires. Whatever the encoding, we assume it is efficient enough so that the following property holds.

\begin{proposition}
On input an $\X$-circuit $\Ci$ and a vector $x$, the output $\Ci(x)$ can be computed in time $\Oh(\card{\Ci})$ (in the usual RAM model of computation). \qed
\end{proposition}
The following proposition records a similar observation for circuit rewriting. 
\begin{proposition}
\label{thm:rewrite-factor-g}
Both \srewrite{} and \xrewrite{} can be solved in time $\Oh(\card{\Ci}^2)$.
\end{proposition}
\begin{proof}
Suppose we are given an \OR-circuit $\Ci$ as input. The matrix $A$ computed by $\Ci$ can be easily extracted from $\Ci$ in time $\Oh(|\Ci|^2)$. We then simply output the trivial depth-1 \SUM-circuit for $A$ that has size at most $n^2 \leq |\Ci|^2$.
\end{proof}

\subsection{Proof of Theorem \ref{thm:rewrite}} \label{sec:proof-of-rewrite}

The main technical ingredient in our proof is Lemma \ref{lemma:covering_to_circuit_rewriting} below, which states that if subquadratic-time rewriting algorithms exist, then certain simple covering problems can be solved faster than in a trivial manner.

In the following we consider set systems defined by $L_1, \dotsc, L_n$ and $R_1, \dotsc, R_n$ that are (not necessarily distinct) subsets of~$[m]$. We say that $(i,j)$ is a \emph{covering pair} if $L_j \cup R_i = [m]$.

\begin{lemma}\label{lemma:covering_to_circuit_rewriting}
Suppose we are given sets $L_1, \dotsc, L_n,R_1, \dotsc, R_n\subseteq [m]$ as input.
\begin{enumerate}[label=(\alph*),leftmargin=*]
\item If \srewrite{} can be solved in time $\Oh(\card{\Ci}^{2-\epsilon})$ for some constant $\epsilon > 0$, then the number of covering pairs can be computed in time $\Oh((nm)^{2-\epsilon})$.
\item If \xrewrite{} can be solved in time $\Oh(\card{\Ci}^{2-\epsilon})$ for some constant $\epsilon > 0$, then the parity of the number of covering pairs can be computed in time $\Oh((nm)^{2-\epsilon})$.
\end{enumerate}
\end{lemma}

\begin{proof}[Proof of (a).]
Let $A=(a_{ij})$ be an $n\times n$ matrix defined by $a_{ij} = 1$ iff $(i,j)$ is a covering pair. We show how to compute $|A|$ without constructing $A$ explicitly.

Suppose for a moment that we had a small \SUM-circuit $\Ci$ for $A$. The value $|A|$ can be recovered from the circuit $\Ci$ in time $\Oh(|\Ci|)$ via the following trick: evaluate $\Ci$ (over the integers) on the all-1 vector $\mathbbm{1}$ to obtain $y=\Ci(\mathbbm{1})\in \N^n$; but now
\begin{equation} \label{eq:first-attempt}
|A| = \mathbbm{1}^\top\!A\mathbbm{1} = \mathbbm{1}^\top \Ci(\mathbbm{1}) = y_1 + \cdots + y_n.
\end{equation}

Unfortunately, we do not know how to construct a small \SUM-circuit for $A$. Instead, our key observation below will be that \emph{the complement matrix} $\bar{A}$ admits an \OR-circuit $\Ci^\OR$ of size only $|\Ci^\OR|=O(nm)$. By assumption, we can then rewrite $\Ci^\OR$ as a \SUM-circuit $\Ci^+$ in time $\Oh(|\Ci^\OR|^{2-\epsilon})=\Oh((nm)^{2-\epsilon})$. In particular, the size of the new circuit must also be
\[
|\Ci^+| = \Oh\left((nm)^{2-\epsilon}\right).
\]
Analogously to (\ref{eq:first-attempt}) we can then recover $|A|$ from $\Ci^\SUM$ in time $\Oh(|\Ci^\SUM|)$:
\[
|A| = n^2 - |\bar{A}| = n^2 - \mathbbm{1}^\top\Ci^\SUM(\mathbbm{1}).
\]
Indeed, it remains to describe how to construct $\Ci^\OR$ for $\bar{A}$ in time $\Oh(nm)$.

\paragraph{Construction}
Define a depth-2 circuit $\Ci^\OR$ follows: The $0$-th layer of $\Ci^\OR$ hosts input gates $l_j$, $j\in[n]$; the $1$-st layer contains intermediate gates $g_k$, $k\in[m]$; and the $2$-nd layer contains output gates $r_i$, $i\in[n]$. Each input gate $l_j$ is connected to gates $g_k$ for $k\in [m]\smallsetminus L_j$; similarly, each output gate $r_i$ is connected to gates $g_k$ for $k\in [m] \smallsetminus R_i$. To see that $\Ci^\OR$ computes $\bar{A}$ note that there is a path from input $l_i$ to output $r_j$ iff there is a $k\in [m]$ such that $k\notin L_i\cup R_j$ iff $(i,j)$ is not a covering pair. Note also that $|\Ci^\OR| \leq 2nm$ and that the construction takes time~$\Oh(nm)$.
\end{proof}

\begin{proof}[Proof of (b).]
The proof is the same as above, except we work over $\GF(2)$.
\end{proof}

Next, we reduce \NumSAT{} and \ParitySAT{} to the covering problems in Lemma~\ref{lemma:covering_to_circuit_rewriting}. Here we are essentially applying a technique of Williams \cite[Theorem~5]{Williams:2csp}.

\begin{theorem}\label{thm:reduction-to-covering} We have the following reductions:
\begin{enumerate}[label=(\alph*),leftmargin=*]
\item If \srewrite{} can be solved in time $\Oh(\card{\Ci}^{2-\epsilon})$ for some $\epsilon > 0$, then \NumSAT{} can be solved in time $2^{(1-\epsilon/2)n}\poly(n,m)$.
\item If \xrewrite{} can be solved in time $\Oh(\card{\Ci}^{2-\epsilon})$ for some $\epsilon > 0$, then \ParitySAT{} can be solved in time $2^{(1-\epsilon/2)n}\poly(n,m)$.
\end{enumerate}
\end{theorem}

\begin{proof}
Let $\varphi = \{ C_1, \dotsc C_m \}$ be an instance of \CNFSAT{} over variables $x_1,\dotsc,x_n$.
Without loss of generality (by inserting
one variable as necessary), we may assume that $n$ is even.
Call the variables $x_1,\ldots,x_{n/2}$ {\em left} variables
and the variables $x_{n/2+1},\ldots,x_{n}$ {\em right} variables.

For each truth assignment $s\in\{0,1\}^{n/2}$ to the left variables, let $L_s \subseteq \varphi$ be the set of clauses satisfied by $s$. Similarly, for assignment $t\in\{0,1\}^{n/2}$ to the right variables, let $R_t\subseteq \varphi$ be the set of clauses satisfied by $t$. Clearly, the compound assignment $(s,t)$ to all the variables satisfies $\varphi$ if and only if $L_s \cup R_t = \varphi$. That is, the number of satisfying assignments is precisely the number of covering pairs of the set system $\{L_s,R_t\}$, $s,t\in\{0,1\}^{n/2}$. Thus, both claims follow from Lemma~\ref{lemma:covering_to_circuit_rewriting}.
\end{proof}
We can now finish the proof of Theorem \ref{thm:rewrite}:
\begin{enumerate}[label=$-$,noitemsep]
\item For \srewrite{} the result follows immediately from Theorem~\ref{thm:reduction-to-covering}.
\item For \xrewrite{} we need to make the following additional argument.
As discussed by Cygan et al.~\cite{cygan2012} the $k$-CNF Isolation Lemma of Calabro et al.~\cite{calabro2008} can be applied to show that any $2^{(1-\epsilon)n}\poly(n,m)$ time algorithm for \ParitySAT{} can be turned into an $2^{(1-\epsilon')n}\poly(n,m)$ time Monte Carlo algorithm for \CNFSAT{} where $\epsilon' > 0$. Recognising this, the result follows from Theorem~\ref{thm:reduction-to-covering}.
\end{enumerate}

\bigskip
\paragraph{\bf Acknowledgements}

We are grateful to Stasys Jukna for pointing out a more direct proof of Corollary~\ref{cor:or-xor-sep} as referenced in the text. We also thank Igor Sergeev for providing many references, in particular, one simplifying our proof of Theorem~\ref{thm:k-uniform}. Furthermore, we thank Jukka Suomela for discussions.

This research is supported in part by Academy of Finland, grants 132380 and 252018 (M.G.), 252083 and 256287 (P.K.), and by Helsinki Doctoral Programme in Computer Science - Advanced Computing and Intelligent Systems (J.K.).

\DeclareUrlCommand{\Doi}{\urlstyle{same}}
\newcommand{\doi}[1]{\href{http://dx.doi.org/#1}{\footnotesize\sf doi:\Doi{#1}}}


\begin{thebibliography}{30}
\providecommand{\natexlab}[1]{#1}
\providecommand{\url}[1]{\texttt{#1}}
\expandafter\ifx\csname urlstyle\endcsname\relax
  \providecommand{\doi}[1]{doi: #1}\else
  \providecommand{\doi}{doi: \begingroup \urlstyle{rm}\Url}\fi

\bibitem[Alon and Spencer(2000)]{alon00probabilistic}
N.~Alon and J.~H. Spencer.
\newblock \emph{The Probabilistic Method}.
\newblock John Wiley \& Sons, 2 edition, 2000.

\bibitem[Alon et~al.(1990)Alon, Karchmer, and Wigderson]{alon90linear}
N.~Alon, M.~Karchmer, and A.~Wigderson.
\newblock Linear circuits over {GF}(2).
\newblock \emph{SIAM Journal on Computing}, 19\penalty0 (6):\penalty0
  1064--1067, 1990.
\newblock \doi{10.1137/0219074}.

\bibitem[Bj\"{o}rklund et~al.(2012)Bj\"{o}rklund, Husfeldt, Kaski, Koivisto,
  Nederlof, and Parviainen]{Bjorklund:zeta}
A.~Bj\"{o}rklund, T.~Husfeldt, P.~Kaski, M.~Koivisto, J.~Nederlof, and
  P.~Parviainen.
\newblock Fast zeta transforms for lattices with few irreducibles.
\newblock In \emph{Proceedings of the 23rd Annual ACM-SIAM Symposium on
  Discrete Algorithms (SODA 2012)}, pages 1436--1444. SIAM, 2012.

\bibitem[Bollob{\'a}s(2001)]{bollobas01random}
B.~Bollob{\'a}s.
\newblock \emph{Random Graphs}.
\newblock Number~73 in Cambridge studies in advanced mathematics. Cambridge
  University Press, 2nd edition, 2001.

\bibitem[Boyar et~al.(2013)Boyar, Matthews, and Peralta]{boyar13logic}
J.~Boyar, P.~Matthews, and R.~Peralta.
\newblock Logic minimization techniques with applications to cryptology.
\newblock \emph{Journal of Cryptology}, 26:\penalty0 280--312, 2013.
\newblock \doi{10.1007/s00145-012-9124-7}.

\bibitem[Calabro et~al.(2008)Calabro, Impagliazzo, Kabanets, and
  Paturi]{calabro2008}
C.~Calabro, R.~Impagliazzo, V.~Kabanets, and R.~Paturi.
\newblock The complexity of unique {$k$}-{SAT}: An isolation lemma for
  {$k$}-{CNF}s.
\newblock \emph{Journal of Computer and System Sciences}, 74\penalty0
  (3):\penalty0 386--393, 2008.
\newblock \doi{10.1016/j.jcss.2007.06.015}.

\bibitem[Cygan et~al.(2012)Cygan, Dell, Lokshtanov, Marx, Nederlof, Okamoto,
  Paturi, Saurabh, and Wahlstrom]{cygan2012}
M.~Cygan, H.~Dell, D.~Lokshtanov, D.~Marx, J.~Nederlof, Y.~Okamoto, R.~Paturi,
  S.~Saurabh, and M.~Wahlstrom.
\newblock On problems as hard as {CNF-SAT}.
\newblock In \emph{Proceedings of the 27th Conference on Computational
  Complexity (CCC 2012)}, pages 74--84. IEEE, 2012.
\newblock \doi{10.1109/CCC.2012.36}.

\bibitem[Find et~al.(2013)Find, G{\"o}{\"o}s, Kaski, and
  Korhonen]{find13separating}
M.~G. Find, M.~G{\"o}{\"o}s, P.~Kaski, and J.~H. Korhonen.
\newblock Separating {OR}, {SUM}, and {XOR} circuits.
\newblock Submitted, 2013.

\bibitem[G\'{a}l et~al.(2012)G\'{a}l, Hansen, Kouck\'{y}, Pudl\'{a}k, and
  Viola]{gal12tight}
A.~G\'{a}l, K.~A. Hansen, M.~Kouck\'{y}, P.~Pudl\'{a}k, and E.~Viola.
\newblock Tight bounds on computing error-correcting codes by bounded-depth
  circuits with arbitrary gates.
\newblock In \emph{Proceedings of the 44th Annual ACM Symposium on Theory of
  Computing (STOC 2012)}, pages 479--494. ACM, 2012.
\newblock \doi{10.1145/2213977.2214023}.

\bibitem[Garey and Johnson(1979)]{garey-johnson}
M.~R. Garey and D.~S. Johnson.
\newblock \emph{{C}omputers and {I}ntractability: {A} {G}uide to the {T}heory
  of {NP}-{C}ompleteness}.
\newblock W.H. Freeman and Company, 1979.

\bibitem[Gashkov and Sergeev(2011)]{gashkov11complexity}
S.~B. Gashkov and I.~S. Sergeev.
\newblock On the complexity of linear {B}oolean operators with thin matrices.
\newblock \emph{Journal of Applied and Industrial Mathematics}, 5:\penalty0
  202--211, 2011.
\newblock \doi{10.1134/S1990478911020074}.

\bibitem[Grinchuk and Sergeev(2011)]{grinchuk11thin}
M.~I. Grinchuk and I.~S. Sergeev.
\newblock Thin circulant matrixes and lower bounds on complexity of some
  {B}oolean operators.
\newblock \emph{Diskretny{\u\i} Analiz i Issledovanie Operatsi{\u\i}},
  18:\penalty0 38--53, 2011.

\bibitem[Impagliazzo and Paturi(2001)]{Impagliazzo:ksat}
R.~Impagliazzo and R.~Paturi.
\newblock On the complexity of {$k$}-{SAT}.
\newblock \emph{Journal of Computer and System Sciences}, 62\penalty0
  (2):\penalty0 367--375, 2001.
\newblock \doi{10.1006/jcss.2000.1727}.

\bibitem[J{\"a}rvisalo et~al.(2012)J{\"a}rvisalo, Kaski, Koivisto, and
  Korhonen]{jarvisalo12finding}
M.~J{\"a}rvisalo, P.~Kaski, M.~Koivisto, and J.~H. Korhonen.
\newblock Finding efficient circuits for ensemble computation.
\newblock In \emph{Proceedings of the 15th International Conference on Theory
  and Applications of Satisfiability Testing (SAT 2012)}, pages 369--382.
  Springer, 2012.
\newblock \doi{10.1007/978-3-642-31612-8_28}.

\bibitem[Jukna(2006)]{jukna06disproving}
S.~Jukna.
\newblock Disproving the single level conjecture.
\newblock \emph{SIAM Journal on Computing}, 36\penalty0 (1):\penalty0 83--98,
  2006.
\newblock \doi{10.1137/S0097539705447001}.

\bibitem[Jukna(2012)]{jukna12boolean}
S.~Jukna.
\newblock \emph{Boolean Function Complexity: Advances and Frontiers}, volume~27
  of \emph{Algorithms and Combinatorics}.
\newblock Springer, 2012.

\bibitem[Jukna(2013)]{jukna13comment}
S.~Jukna.
\newblock Comment on {XOR} versus {OR} circuits, April 2013.
\newblock URL
  \url{http://www.thi.informatik.uni-frankfurt.de/~jukna/boolean/comment9.html}.

\bibitem[Knuth(1998)]{knuth98art}
D.~E. Knuth.
\newblock \emph{The Art of Computer Programming}, volume~2.
\newblock Addison--Wesley, 3rd edition, 1998.

\bibitem[Koll{\'a}r et~al.(1996)Koll{\'a}r, R{\'o}nyai, and
  Szab{\'o}]{kollar96norm}
J.~Koll{\'a}r, L.~R{\'o}nyai, and T.~Szab{\'o}.
\newblock Norm-graphs and bipartite {T}ur{\'a}n numbers.
\newblock \emph{Combinatorica}, 16\penalty0 (3):\penalty0 399--406, 1996.
\newblock \doi{10.1007/BF01261323}.

\bibitem[Kushilevitz and Nisan(1997)]{kushilevitz97communication}
E.~Kushilevitz and N.~Nisan.
\newblock \emph{Communication Complexity}.
\newblock Cambridge University Press, 1997.

\bibitem[Lamagna and Savage(1974)]{lamagna74computational}
E.~A. Lamagna and J.~E. Savage.
\newblock Computational complexity of some monotone functions.
\newblock In \emph{IEEE Conference Record of 15th Annual Symposium on Switching
  and Automata Theory}, pages 140--144, 1974.
\newblock \doi{10.1109/SWAT.1974.9}.

\bibitem[Mehlhorn(1979)]{mehlhorn79some}
K.~Mehlhorn.
\newblock Some remarks on {B}oolean sums.
\newblock \emph{Acta Informatica}, 12:\penalty0 371--375, 1979.
\newblock \doi{10.1007/BF00268321}.

\bibitem[Miltersen(1998)]{miltersen98error}
P.~B. Miltersen.
\newblock Error correcting codes, perfect hashing circuits, and deterministic
  dynamic dictionaries.
\newblock In \emph{Proceedings of the 9th Annual ACM-SIAM Symposium on Discrete
  Algorithms (SODA 1998)}, pages 556--563. SIAM, 1998.

\bibitem[Nechiporuk(1971)]{nechiporuk71boolean}
{\'E}.~I. Nechiporuk.
\newblock On a {B}oolean matrix.
\newblock \emph{Systems Theory Research}, 21:\penalty0 236--239, 1971.

\bibitem[P{\v{a}}tra{\c{s}}cu and Williams(2010)]{DBLP:conf/soda/PatrascuW10}
M.~P{\v{a}}tra{\c{s}}cu and R.~Williams.
\newblock On the possibility of faster {SAT} algorithms.
\newblock In \emph{Proceedings of the 21st Annual ACM-SIAM Symposium on
  Discrete Algorithms (SODA 2010)}, pages 1065--1075. SIAM, 2010.

\bibitem[Pippenger(1980{\natexlab{a}})]{pippenger80another}
N.~Pippenger.
\newblock On another {B}oolean matrix.
\newblock \emph{Theoretical Computer Science}, 11\penalty0 (1):\penalty0
  49--56, 1980{\natexlab{a}}.
\newblock \doi{10.1016/0304-3975(80)90034-1}.

\bibitem[Pippenger(1980{\natexlab{b}})]{pippenger80evaluation}
N.~Pippenger.
\newblock On the evaluation of powers and monomials.
\newblock \emph{SIAM Journal on Computing}, 9\penalty0 (2):\penalty0 230--250,
  1980{\natexlab{b}}.
\newblock \doi{10.1137/0209022}.

\bibitem[Pudl{\'a}k and R{\"o}dl(2004)]{pudlak04pseudorandom}
P.~Pudl{\'a}k and V.~R{\"o}dl.
\newblock Pseudorandom sets and explicit constructions of {R}amsey graphs.
\newblock In \emph{Complexity of computations and proofs}, volume~13 of
  \emph{Quaderni Di Matematica}. 2004.

\bibitem[Spielman(1996)]{spielman96linear}
D.~A. Spielman.
\newblock Linear-time encodable and decodable error-correcting codes.
\newblock \emph{IEEE Transactions on Information Theory}, 42\penalty0
  (6):\penalty0 1723--1731, 1996.
\newblock \doi{10.1109/18.556668}.

\bibitem[Williams(2005)]{Williams:2csp}
R.~Williams.
\newblock A new algorithm for optimal 2-constraint satisfaction and its
  implications.
\newblock \emph{Theoretical Computer Science}, 348\penalty0 (2--3):\penalty0
  357--365, 2005.
\newblock \doi{10.1016/j.tcs.2005.09.023}.

\end{thebibliography}
\end{document}